\UseRawInputEncoding

\documentclass[10pt]{article}

\usepackage{hyperref}
\usepackage{amsthm,amsmath,amssymb}
\usepackage{enumerate}

\usepackage{fullpage}
\newtheorem{theorem}{Theorem}
\newtheorem{lemma}[theorem]{Lemma}
\newtheorem{remark}[theorem]{Remark}
\newtheorem{definition}[theorem]{Definition}
\newtheorem{corollary}[theorem]{Corollary}

\newtheorem{proposition}[theorem]{Proposition}
\newtheorem{example}[theorem]{Example}

\newcommand{\Cov}{\operatorname{Cov}}
\newcommand{\diag}{\operatorname{diag}}

\newcommand{\pd }{\mathbf{S}^+}

\newcommand{\psd }{\mathbf{S}_0^+}
\newcommand{\sym }{\mathbf{S}}
\newcommand{\Tr}{\operatorname{Tr}}

\newcommand{\id}{\operatorname{id}}

\newcommand{\EE}{\mathbb{E}}

\renewcommand{\tilde}{\widetilde}
\newcommand{\ospan}{\operatorname{span}}
\newcommand{\ocov}{\operatorname{Cov}}

\title{ Equality cases in the Anantharam--Jog--Nair inequality}
\author{Efe Aras, Thomas A.~Courtade and Albert Zhang\\University of California, Berkeley}
\date{~}

\begin{document}

\maketitle

\begin{abstract}
Anantharam, Jog and Nair recently unified the Shannon--Stam inequality and the entropic form of the Brascamp--Lieb inequalities under a common inequality.  They left open the problems of extremizability and  characterization of extremizers.  Both questions are resolved in the present paper.
\end{abstract}

\section{Preliminaries}
We begin by briefly fixing notation and definitions that will be needed throughout.   A Euclidean space $E$  is a finite-dimensional Hilbert space over the real field, equipped with Lebesgue measure.  For a probability measure $\mu$ on $E$, absolutely continuous with respect to Lebesgue measure, and a random vector $X\sim \mu$, we define the Shannon entropy
$$
h(X) \equiv h(\mu) :=-\int_E \log\left( \frac{d\mu}{dx}\right)d\mu, 
$$
provided the integral exists.  If $\mu$ is not absolutely continuous with respect to Lebesgue measure, we adopt the convention that $h(\mu) := -\infty$.  We let $\mathcal{P}(E)$ denote the set of  probability measures on $E$ having finite entropies and second moments.  When there is no cause for ambiguity, we adopt the usual notational convention where a random vector $X$ and its law $\mu$ are used interchangeably.  So, for example, writing $X\in \mathcal{P}(E)$ means that $X$ is a random vector taking values in $E$, having finite entropy and finite second moments.  

For $x,y\in E$, we denote the standard (Euclidean) inner product as $x^T y$, and denote the Euclidean metric by $|\cdot|$ (i.e., $|x| := \sqrt{x^Tx}$).  If $A : E\to  E'$ is a linear map between Euclidean spaces $E,E'$, we let $A^T:E'\to E$ denote its adjoint satisfying
$$
(Ax)^T y = x^T (A^T y),~~~\forall x\in E, y\in E'.  
$$
All of this notation is consistent with the  representation of linear maps as matrices.   We let $\sym(E)$ denote the set of symmetric linear maps from $E$ to itself (i.e., $A\in \sym(E)$ iff $A=A^T$), and $\pd(E)$ denote the subset of positive definite linear maps (i.e., $A\in \pd(E)$ iff $A=A^T$  and $x^T A x>0$ for all nonzero $x\in E$). 

For a random vector $X\sim \mu \in \mathcal{P}(E)$, its covariance is defined as the (positive semidefinite) symmetric linear map
$$
\ocov(X) = \int_E (x-\EE[X]) (x-\EE[X])^T d\mu(x) \in \sym(E),
$$
where $\EE$ denotes expectation (here, with respect to $\mu$).  The Gaussian distribution on $E$ with mean $m$ and covariance $\Sigma\in \pd(E)$ is denoted by $N(m,\Sigma)$.  A Gaussian random vector $X$ is said to be {\bf isotropic} if it has covariance proportional to the identity map.  The standard Gaussian distribution on $E$ is denoted by $\gamma_E$.  

Of course, all Euclidean spaces $E,E'$ of dimensions $m$ and $n$, respectively, can always be identified as $\mathbb{R}^m$ and $\mathbb{R}^n$, respectively.  Moreover, any linear transformation $A: E\to E'$ can be expressed as a real $n\times m$   matrix.  Our notation is chosen to be compatible with this, but for various reasons it is notationally more convenient to state things abstractly.  For example, this avoids ambiguity that can result from referring to two different Euclidean spaces of the same dimension.

Throughout, we consider collections of Euclidean spaces $(E_i)_{i=1}^k$, $(E^j)_{j=1}^m$, and corresponding sets of positive real numbers $\mathbf{c} = (c_i)_{i=1}^k$ and $\mathbf{d} = (d_j)_{j=1}^m$.  A {\bf datum} is a triplet $(\mathbf{c},\mathbf{d},\mathbf{B})$ where $\mathbf{B}=(B_j)_{j=1}^m$ is a collection of linear maps $B_j : E_0 \to E^j$, with common domain $E_0 := \oplus_{i=1}^k E_i$.  Given the structure of $E_0$, we  let $\pi_{E_i}: E_0 \to E_i$ denote the coordinate projections.  A vector $x\in E_0$ will frequently be written in its coordinate representation $x = (x_1, \dots, x_k)$, where $x_i = \pi_{E_i}(x)$, $1\leq i\leq k$.    If $A_i : E_i \to E_i$, $1\leq i \leq k$, are linear maps, then the direct sum of operators $A =    \oplus_{i=1}^k A_i$ is a linear map from $E_0$ to itself and, without confusion, can be denoted as the block-diagonal operator
$$
A = \operatorname{diag}(A_1, \dots, A_k).
$$
For a set $V$, we  let $\id_{V}: V\to V$ denote the identity map from $V$ to itself.  So, as an example of the above, we have  $\id_{E_0} = \oplus_{i=1}^k \id_{E_i} \equiv  \operatorname{diag}(\id_{E_1}, \dots, \id_{E_k})$.  Again, this is all compatible with the representation of linear operators as matrices. 

We conclude this section by recording a few associated definitions for convenience.
\begin{definition}
A subspace $T\subset E_0$ is said to be {\bf product-form} if it can be written as $T = \oplus_{i=1}^k T_i$, where $T_i \subset E_i$ for $1\leq i\leq k$. 
\end{definition}

\begin{definition}   A subspace $T \subset E_0$ is said to be {\bf critical} for   $(\mathbf{c},\mathbf{d},\mathbf{B})$  if it is  product-form, and 
 \begin{align}
\sum_{i=1}^k c_i \dim(\pi_{E_i}T ) = \sum_{j=1}^m d_j \dim(B_j T). \notag 
\end{align}
\end{definition}

\begin{definition}\label{def:EquivalentData}
Two data $(\mathbf{c},\mathbf{d},\mathbf{B})$ and $(\mathbf{c'},\mathbf{d'},\mathbf{B'})$ are said to be {\bf equivalent} if $\mathbf{c}=\mathbf{c'}$, $\mathbf{d}=\mathbf{d'}$, and there exist invertible linear transformations $A_j : E^j \to E^j$  and $C_i : E_i \to E_i$ such that 
\begin{align}
B'_j = A_j^{-1}B_j C^{-1} \hspace{5mm}\mbox{for each $1\leq j \leq m$}, \label{eq:equivalentDatum}
\end{align}
where $C:= \operatorname{diag}(C_1, \dots, C_k)$.
\end{definition}

We remark that, in the special case of $k=1$, the definitions of critical subspaces and equivalent data coincide with those found in \cite{bennett2008brascamp}.  For general $k$, all three definitions coincide with those in \cite{CourtadeLiu21}.

\section{The Anantharam--Jog--Nair inequality}
  For a datum $(\mathbf{c},\mathbf{d},\mathbf{B})$, Anantharam, Jog and Nair (AJN) characterized the  best (i.e., smallest) constant $C$ such that the entropy inequality  
\begin{align}
\sum_{i=1}^k c_i h(X_i) \leq  \sum_{j=1}^m d_j h(B_j X) + C \label{eq:multimarginalindependent}
\end{align}
holds for any choice of independent random vectors $X_i   \in \mathcal{P}(E_i)$, $1\leq i \leq k$, with $X:=(X_1, \dots, X_k)$.  This inequality unifies the Shannon--Stam inequality \cite{shannon48, stam59} and the entropic formulation of the (Euclidean) Brascamp--Lieb inequalities \cite{carlen2009subadditivity, brascamp1976best} under a common framework.  Extending the Gaussian saturation properties enjoyed by each (see, e.g., \cite{CarlenSoffer} and \cite{lieb1990gaussian}), Anantharam, Jog and Nair  showed that the best constant can be computed by considering only Gaussian $X_i$'s, and gave necessary and sufficient conditions for finiteness.   More precisely, their main result is the following:
\begin{theorem}[AJN inequality \cite{anantharam2019unifying}]\label{thm:AJNinequality} Fix a datum $(\mathbf{c},\mathbf{d},\mathbf{B})$. 
For any random vectors $X_i \in \mathcal{P}(E_i)$, $1\leq i \leq k$ and  $X = (X_1,\dots, X_k)$, 
\begin{align} 
\sum_{i=1}^k c_i h(X_i) - \sum_{j=1}^m d_j h(B_j X) \leq  C_g(\mathbf{c},\mathbf{d},\mathbf{B}), \label{eq:AJNinequality}
\end{align}
where $C_g(\mathbf{c},\mathbf{d},\mathbf{B})$ is defined as the supremum of the LHS over independent Gaussian vectors $(X_i)_{i=1}^k$.  Moreover, the  constant $C_g(\mathbf{c},\mathbf{d},\mathbf{B})$ is finite if and only if the following two conditions are satisfied.
\begin{enumerate}[(i)]
\item {\bf Scaling condition:}  It holds that
\begin{align}
\sum_{i=1}^k c_i \dim(E_i) = \sum_{j=1}^m d_j \dim(E^j). \label{eq:ScalingCond}
\end{align}
\item{\bf Dimension condition:}    For all product-form subspaces $T  \subset E_0$,  
 \begin{align}
\sum_{i=1}^k c_i \dim(\pi_{E_i} T ) \leq \sum_{j=1}^m d_j \dim(B_j T).\label{eq:DimCond}
\end{align}
\end{enumerate}
\end{theorem}

Anantharam, Jog and Nair left open the question of extremizability. That is, when do there exist random vectors $(X_i)_{i=1}^k$ such that \eqref{eq:AJNinequality} is met with equality, and  what form do any such extremizers take?  The goal of this paper is to answer both questions completely.   The first question is addressed in Section \ref{sec:geom}, and the second in Section \ref{sec:Extremizers}.

The precise characterization of extremizers is somewhat complicated, but the general idea is easily understood in the context of a toy example.  For $\lambda \in (0,1)$, the following   holds:  If $(X,Y)$ is independent of $Z$, and  $Y$ and $Z$ are of the same dimension, then 
\begin{align}
\lambda h(X, Y) + (1-\lambda)  h(Z) \leq \lambda h(X) + h( \lambda^{1/2} Y + (1-\lambda)^{1/2}Z ) .\label{simpleAJN}
\end{align}
This inequality is obtained by a concatenation of subadditivity of entropy   and the Shannon--Stam inequality.  
Restricting attention to cases where all entropies are finite, we can use known equality cases for both to assert that $(X,Y)$ and $Z$ are extremizers in \eqref{simpleAJN} if and only if (i) $X$ and $Y$ are independent; and (ii) $Y$ and $Z$ are Gaussian with identical covariances.

Roughly speaking, all extremizers of the AJN inequality \eqref{eq:AJNinequality} resemble the above example.  That is, extremizers are characterized by a rigid factorization into independent components, where some components can have any distribution, and    the remaining are necessarily Gaussian with covariances that are typically linked in some way.  

Our approach  leverages an assemblage of techniques developed by various researchers.  In particular, the question of extremizability is addressed by identifying a suitable notion of ``AJN-geometricity'',  and showing that all extremizable data are equivalent to AJN-geometric data. This parallels the approach developed by Bennett, Carbery, Christ  and Tao \cite{bennett2008brascamp} for the functional form of the Brascamp--Lieb  inequalities, which by duality \cite{carlen2009subadditivity} can be realized as an instance of  \eqref{eq:AJNinequality}.  The Gaussian saturation property of AJN-geometric data is established by a stochastic argument involving the F\"ollmer drift (see Appendix \ref{sec:FollmersDrift} for definitions and properties), inspired by Lehec's stochastic proof of the Shannon--Stam  inequality \cite{lehec2013representation}.  This stochastic proof lends itself to identifying the structure of extremizers (when they exist), by combining key ideas from Valdimarsson's characterization of optimizers in the functional Brascamp--Lieb  inequalities \cite{valdimarsson08} together with tools from Eldan and Mikulincer's work on the stability of the Shannon--Stam inequality \cite{EldanMikulincer}.

\section{Extremizability and Geometricity}\label{sec:geom}

We first address the question of  when \eqref{eq:AJNinequality} is extremizable.  To make things precise, we say that a datum $(\mathbf{c},\mathbf{d},\mathbf{B})$ is {\bf extremizable} if $C_g(\mathbf{c},\mathbf{d},\mathbf{B})$ is finite and there exist independent $X_i \in \mathcal{P}(E_i)$, $1\leq i \leq k$ such that \eqref{eq:AJNinequality} is met with equality.  We say that $(\mathbf{c},\mathbf{d},\mathbf{B})$ is {\bf Gaussian-extremizable} if  $C_g(\mathbf{c},\mathbf{d},\mathbf{B})$ is finite and  there exist independent Gaussian $(X_i)_{i=1}^k$ meeting  \eqref{eq:AJNinequality}  with equality.

In analogy to definitions made in the context of Brascamp--Lieb inequalities, we define the class of {AJN-geometric} data below.  Their significance to  \eqref{eq:AJNinequality} is the same as that of geometric data to inequalities of Brascamp--Lieb-type.  In particular, we will see that all (Gaussian-)extremizable instances of \eqref{eq:AJNinequality} are equivalent to AJN-geometric data.

\begin{definition}[AJN-Geometric datum]
A datum $(\mathbf{c},\mathbf{d},\mathbf{B})$ is said to be {\bf AJN-geometric} if  
\begin{enumerate}[(i)]
\item $B_j  B_j^T = \id_{E^j}$ for each $1\leq j\leq m$; and 
\item we have the operator identity
 \begin{align}
\sum_{j=1}^m d_j \pi_{E_i} B^T_j  B_j \pi^T_{E_i} = c_i \id_{E_i}, \hspace{5mm}\mbox{for each~}1\leq i\leq k.  \label{eq:AJNGeometricOperatorIneq}
\end{align}
\end{enumerate}
\end{definition}
\begin{remark}\label{rmk:scaling}
Conditions (i)-(ii) together imply  the scaling condition \eqref{eq:ScalingCond}. This can be seen by taking traces in \eqref{eq:AJNGeometricOperatorIneq}, summing from $i=1,\dots, k$, and using the cyclic and linearity properties of trace together with (ii). 
\end{remark}

AJN-geometric data  have the convenient property that $C_g(\mathbf{c},\mathbf{d},\mathbf{B})=0$, and they are extremizable by standard Gaussians.  We summarize as a formal proposition. 
\begin{proposition}\label{prop:AJNGeomCgEquals0}
If $(\mathbf{c},\mathbf{d},\mathbf{B})$ is AJN-geometric,  then $C_g(\mathbf{c},\mathbf{d},\mathbf{B})=0$ and $X\sim N(0,\id_{E_0})$ achieves equality in \eqref{eq:AJNinequality}.
\end{proposition}
\begin{proof}
We'll use the properties of the F\"ollmer drift summarized in Appendix \ref{sec:FollmersDrift}. 
Begin by fixing  {centered} $\mu_i\in \mathcal{P}(E_i)$, $1\leq i \leq k$, and let $(W_t)_{t\geq 0}$  be a Brownian motion on $E_0$ with $\operatorname{Cov}(W_1) = \id_{E_0}$.  By Theorem \ref{thm:existsDrifts} and \eqref{eq:FollmerMean}, there is a drift $U_t = \int_0^t u_s ds$  such that $\EE[u_t]=0$ and $(\pi_{E_i}(u_t))_{i=1}^k$ are independent for all $0\leq t\leq 1$,
\begin{align}
  (W_1 + U_1) \sim \mu_1 \otimes \cdots \otimes \mu_k, \label{rightTargetDist}
\end{align}
and $D(\mu_i\|\gamma_{E_i}) = \frac{1}{2} \int_{0}^1\EE | \pi_{E_i}(u_s) |^2 ds$ for each $1\leq i\leq k$. Therefore, 
\begin{align}
\sum_{i=1}^kc_i D(\mu_i \|\gamma_{E_i}) 
&= \frac{1}{2} \EE  \int_0^1 \sum_{i=1}^k c_i    | \pi_{E_i}(u_s) |^2 ds \notag\\
&=   \frac{1}{2} \EE  \int_0^1  \sum_{j=1}^m d_j  |B_j   {u}_s|^2 ds   \label{eq:useOpIn2}   \\
&\geq  \sum_{j=1}^md_j D(B_j \sharp ( \mu_1 \otimes \cdots \otimes \mu_k) \|\gamma_{E^j})   \label{eq:useFollmer2},
\end{align}
where   \eqref{eq:useOpIn2} follows from \eqref{eq:AJNGeometricOperatorIneq} and the properties of $u_t$, and \eqref{eq:useFollmer2} follows from \eqref{rightTargetDist} and  Proposition \ref{prop:FollmerLowerBound} (with construction \eqref{eq:bridgeConstruction}) because  $B_j W_1 \sim \gamma_{E^j}$, due to $B_j B_j^T = \id_{E^j}$ by assumption.  Now, expanding the relative entropies in terms of Shannon entropies and second moments, the second-moment terms cancel due to independence and \eqref{eq:AJNGeometricOperatorIneq}, giving
\begin{align}
\sum_{i=1}^kc_i   h(X_i)   \leq  \sum_{j=1}^m d_j   h(B_j X)   
\label{hIneqAJNGeom}
\end{align}
for any $X_i\sim \mu_i \in \mathcal{P}(E_i)$ and $X\sim \otimes_{i=1}^k \mu_i$, where the centering assumption can  be removed due to translation invariance of Shannon entropy.  The fact that $X\sim \gamma_{E_0}$ is an extremizer follows immediately  from the scaling condition \eqref{eq:ScalingCond} (see Remark \ref{rmk:scaling}) and  the observation that $B_j X\sim \gamma_{E^j}$ (since  $B_j  B_j^T = \id_{E^j}$).
\end{proof}
\begin{remark}
In the case where the datum is such that \eqref{eq:AJNinequality} coincides with the Shannon--Stam inequality, the above proof reduces to that of Lehec \cite{lehec2013representation}.  The new idea  is identifying and incorporating the ``correct" definition  of AJN-geometricity.  When $k=1$, the AJN inequality \eqref{eq:AJNinequality} coincides with the entropic form of the Brascamp--Lieb inequalities, and the definition of AJN-geometricity reduces to the the definition of geometricity for Brascamp--Lieb data found in \cite{bennett2008brascamp}.
\end{remark}

AJN-geometric data have a relatively straightforward geometric interpretation.  In particular, first note that each $E_i$ has a natural isometric embedding into $E_0$ via the inclusion $\pi^T_{E_i}  : E_i \to E_0$. If $(\mathbf{c},\mathbf{d},\mathbf{B})$ is AJN-geometric then $B_j  B_j^T = \id_{E^j}$, which means that each $E^j$ can be isometrically embedded into $E_0$ by the map $B_j^T : E^j \to E_0$.    In this way, we can consider $(E_i)_{i=1}^k$ and $(E^j)_{j=1}^m$ to be subspaces of $E_0$, and $\Pi_{E_i}  := \pi^T_{E_i}\pi_{E_i}$ and $\Pi_{E^j}  := B_j^T B_j$ define the orthogonal projections of $E_0$ onto $E_i$ and $E^j$, respectively.  Thus, the geometric instances of the AJN inequality \eqref{eq:AJNinequality} can be restated in a way that dispenses with the specific linear maps $\mathbf{B}$ as follows.
\begin{corollary}\label{cor:AJNgeoAlt}
Let $E^1,\dots, E^m$ be  subspaces of $E_0 = \oplus_{i=1}^k E_i$. If $\mathbf{c}$ and   $\mathbf{d}$ satisfy
 \begin{align}
\sum_{j=1}^m d_j \Pi_{E_i} \Pi_{E^j}  \Pi_{E_i}  = c_i \Pi_{E_i} , \hspace{5mm}\mbox{for each~}1\leq i\leq k,  \label{eq:AJNGeometricOperatorIneq2}
\end{align}
then for any independent $X_i \in \mathcal{P}(E_i)$, $1\leq i \leq k$, and $X = (X_1,\dots, X_m)$, 
\begin{align} 
\sum_{i=1}^k c_i h( \Pi_{E_i} X) \leq \sum_{j=1}^m d_j h( \Pi_{E^j} X). \label{eq:AJNinequalityGeometricStatement}
\end{align}
Equality is achieved for $X\sim N(0,\id_{E_0})$.
\end{corollary}
\begin{remark} Entropies in \eqref{eq:AJNinequalityGeometricStatement} are computed with respect to Lebesgue measure on the subspace being projected upon.  In particular, we have $h( \Pi_{E_i} X) =h(X_i)$, but have chosen to write \eqref{eq:AJNinequalityGeometricStatement} in a way to emphasize the symmetry of the inequality.
\end{remark}

With the above definitions in hand, the following completely characterizes the (Gaussian-)extremizable instances of Theorem \ref{thm:AJNinequality}.  It is the main result of this section, and specializes to  the extremizability results  in \cite{bennett2008brascamp} for the Brascamp--Lieb functional inequalities when $k=1$. 
\begin{theorem}\label{thm:extremizableIFFgeometricAJN}
The following are equivalent:
\begin{enumerate}[(i)]
\item $(\mathbf{c},\mathbf{d},\mathbf{B})$ is  extremizable.
\item $(\mathbf{c},\mathbf{d},\mathbf{B})$ is Gaussian-extremizable.
\item There  are $K_i \in \mathbf{S}^+(E_i)$, $1\leq i \leq k$,   satisfying
\begin{align}
\sum_{j=1}^m d_j  \pi_{E_i}  B^T_j (B_j K B_j^T)^{-1} B_j \pi^T_{E_i}=  c_i K_i^{-1}, \hspace{5mm}1\leq i\leq k,\label{eq:AJNExtremizabilityCond}
\end{align}
where $K := \operatorname{diag}(K_1 , \dots, K_k )$.
\item $(\mathbf{c},\mathbf{d},\mathbf{B})$ is equivalent to an AJN-geometric datum.
\end{enumerate}
\end{theorem}
\begin{remark}\label{rmk:Kstructure}
For $(K_i)_{i=1}^k$ satisfying \eqref{eq:AJNExtremizabilityCond}, the Gaussians $X_i \sim N(0,K_i)$, $1\leq i \leq k$  are extremal in \eqref{eq:AJNinequality}.  In fact, the proof of Theorem \ref{thm:extremizableIFFgeometricAJN} will show that   if $X_i \in \mathcal{P}(E_i)$, $1\leq i \leq k$ are extremal in \eqref{eq:AJNinequality}, then the covariances $K_i = \ocov(X_i)$ necessarily satisfy \eqref{eq:AJNExtremizabilityCond}.
\end{remark}

As a preliminary observation, we note that the extremizers in  \eqref{eq:AJNinequality} are closed under convolutions. This fact can be extracted from the doubling argument in \cite{anantharam2019unifying}; we state and prove it here for completeness. 
\begin{proposition}\label{prop:AJNextremizersClosedUnderConvolutions}
Fix a datum $(\mathbf{c},\mathbf{d},\mathbf{B})$ that is  extremizable for the AJN inequality \eqref{eq:AJNinequality}.  Let $X=(X_1, \dots, X_k)$ and $Y=(Y_1, \dots, Y_k)$  each satisfy \eqref{eq:AJNinequality} with equality.  If $X,Y$ are independent, then $X+Y = (X_1+Y_1, \dots, X_k+Y_k)$ also satisfies \eqref{eq:AJNinequality} with equality. 
\end{proposition}
\begin{proof}
Define $Z^+ = (Z_1^+, \dots, Z_k^+)$ and $Z^- = (Z_1^-, \dots, Z_k^-)$, where
$$
Z_i^+ := \frac{1}{\sqrt{2}}(X_i + Y_i), \hspace{5mm} Z_i^- := \frac{1}{\sqrt{2}}(X_i - Y_i), \hspace{5mm}1\leq i \leq k.
$$
Observe that
\begin{align}
\sum_{i=1}^k c_i  (h(X_i)  + h(Y_i) ) 
&=  
\sum_{i=1}^k c_i  h(X_i, Y_i) \label{eq:ind1}\\
&=\sum_{i=1}^k c_i \left( h(Z_i^+) + h(Z_i^-|Z_i^+)\right) \label{eq:rotate1}\\
&\leq\sum_{j=1}^m d_j \left( h(B_j Z^+) + h(B_j Z^-|Z^+)\right) +  2C_g(\mathbf{c},\mathbf{d},\mathbf{B}) \label{eq:applyAJN1} \\
&\leq\sum_{j=1}^m d_j \left( h(B_j Z^+) + h(B_j Z^-|B_j Z^+)\right) +2C_g(\mathbf{c},\mathbf{d},\mathbf{B})\label{eq:condReduces1} \\
&=\sum_{j=1}^m d_j \left( h(B_j X, B_j Y)  \right)+2C_g(\mathbf{c},\mathbf{d},\mathbf{B}) \label{eq:orthRotation1} \\
&=\sum_{j=1}^m d_j  ( h(B_j X) +   h(B_j Y) )+2C_g(\mathbf{c},\mathbf{d},\mathbf{B}) \label{eq:ind2}.
\end{align}
In the above, \eqref{eq:ind1} is due to independence; \eqref{eq:rotate1} follows due to orthogonality of the transformation $(X_i,Y_i) \to (Z_i^+,Z_i^-)$ and the chain rule; \eqref{eq:applyAJN1} is two applications of \eqref{eq:AJNinequality}; \eqref{eq:condReduces1} follows because conditioning reduces entropy; \eqref{eq:orthRotation1} is due to the chain rule and orthogonality of the transformation $(B_j Z^+, B_jZ^-) \to (B_j X, B_j Y)$; \eqref{eq:ind2} is again due to independence. 

Since $X$ and $Y$ are extremal by assumption, we have equality throughout.  This implies $Z^+$ is also extremal, and hence we conclude $X+Y$ is extremal by the scaling condition \eqref{eq:ScalingCond}.  \end{proof}

\begin{proof}[Proof of Theorem \ref{thm:extremizableIFFgeometricAJN}] $(i)\Rightarrow(ii)$:  Let $X$ be an extremizer in \eqref{eq:AJNinequality}, and put $Z_n := n^{-1/2} \sum_{\ell=1}^n X^{(i)}$, where $X^{(1)},X^{(2)}, \dots$ are i.i.d.~copies of $X$, which we assume to be zero-mean without loss of generality.  By an application of Proposition \ref{prop:AJNextremizersClosedUnderConvolutions} and the scaling condition  \eqref{eq:ScalingCond} (which holds by finiteness of $ C_g(\mathbf{c},\mathbf{d},\mathbf{B})$), we have that $Z_n$ is an extremizer in   \eqref{eq:AJNinequality} for all $n\geq 1$.    By an application of  the entropic central limit theorem \cite{barronCLT, CarlenSoffer}, it follows that  $Z\sim N(0,\ocov(X))$ is also an extremizer in   \eqref{eq:AJNinequality}.

\smallskip

\noindent$(ii)\Rightarrow(i)$:  This follows immediately from Theorem \ref{thm:AJNinequality}. 

\smallskip

\noindent$(ii)\Rightarrow(iii)$: If $(\mathbf{c},\mathbf{d},\mathbf{B})$ is Gaussian-extremizable, then  there exist $K^*_i \in \mathbf{S}^+(E_i)$, $1\leq i \leq k$ which maximize
$$
(K_i)_{i=1}^k \mapsto \sum_{i=1}^k c_i \log \det(K_i) -  \sum_{j=1}^m d_j \log \det(B_j K B_j^T), 
$$
where $K:= \diag(K_1, \dots, K_k)$ (note  this implies $B_j K^* B_j^T$ is invertible for each $1\leq j \leq m$).  This means, for any index $i$ and any $A_i \in \sym(E_i)$,  we can consider the perturbation $K_i = K_i^* + \epsilon A_i$ for $\epsilon$ sufficiently small, and the function value cannot increase.  By first-order Taylor expansion, this implies
\begin{align*}
c_i \langle A_i, (K^*_i)^{-1} \rangle  &= \sum_{j=1}^m d_j \langle   B_j \pi_{E_i}^T A_i \pi_{E_i} B_j^T, (B_j K^* B_j^T)^{-1} \rangle  \\
&=   \Big\langle   A_i  , \sum_{j=1}^m d_j  \pi_{E_i}  B^T_j (B_j K^* B_j^T)^{-1} B_j \pi^T_{E_i} \Big\rangle , 
\end{align*}
where $ \langle  \cdot, \cdot \rangle $ is the Hilbert--Schmidt (trace) inner product. 
By arbitrariness of $A_i$, we conclude \eqref{eq:AJNExtremizabilityCond}.

\smallskip

\noindent$(iii)\Rightarrow(iv)$:  Let $K$ be as in \eqref{eq:AJNExtremizabilityCond}. The equivalent datum $(\mathbf{c},\mathbf{d},\mathbf{B'})$ defined by 
 $$
 B_j' = (B_j K B_j^T)^{-1/2}B_j K^{1/2} , ~~1\leq j\leq m
 $$
is AJN-geometric. Indeed, $B_j' B_j'^T = \id_{E^j}$ and \eqref{eq:AJNExtremizabilityCond} gives
 \begin{align*}
\sum_{j=1}^m d_j \pi_{E_i} B_j'^T  B_j' \pi^T_{E_i} = \sum_{j=1}^m d_j K_i^{1/2}\pi_{E_i} B_j  (B_j K B_j^T)^{-1}  B_j \pi^T_{E_i} K_i^{1/2} = c_i \id_{E_i}.
\end{align*}

\smallskip

\noindent$(iv)\Rightarrow(ii)$:  Let $(\mathbf{c},\mathbf{d},\mathbf{B'})$ be the geometric datum equivalent to $(\mathbf{c},\mathbf{d},\mathbf{B})$.  In the notation of \eqref{eq:equivalentDatum}, for any $X_i \in \mathcal{P}(E_i)$, $1\leq i \leq k$ and  $X = (X_1,\dots, X_k)$, we have by a change of variables 
\begin{align*} 
&\sum_{i=1}^k c_i h(X_i) - \sum_{j=1}^m d_j h(B_j X) \\
&=  \sum_{i=1}^k c_i h(C_i X_i)  - \sum_{i=1}^k c_i \log\det(C_i) - \sum_{j=1}^m d_j h(B_j' C X) -\sum_{j=1}^m d_j\log\det(A_j) \\
&= \sum_{i=1}^k c_i h(Y_i)  - \sum_{j=1}^m d_j h(B_j' Y) - \sum_{i=1}^k c_i \log\det(C_i) -\sum_{j=1}^m d_j\log\det(A_j),
\end{align*}
where we have defined $Y_i := C_i X_i$, and $Y = (Y_1, \dots, Y_k)$.  Since each $C_i$ is invertible, it is clear that $X$ is a (Gaussan-)extremizer for $(\mathbf{c},\mathbf{d},\mathbf{B})$ if and only if $Y$ is a (Gaussan-)extremizer for $(\mathbf{c},\mathbf{d},\mathbf{B'})$.  The latter is Gaussian-extremizable by the assumption of geometricity and Proposition \ref{prop:AJNGeomCgEquals0}, so the claim follows. 
\end{proof}

\begin{remark}\label{rmk:FRBLrelated}
We remark that Theorem \ref{thm:AJNinequality} can be derived as a limiting case of the  forward-reverse Brascamp--Lieb inequalities \cite{liu2018forward}; details can be found in \cite[Section 4]{CourtadeLiu21}.  There is a counterpart notion of geometricity for the forward-reverse Brascamp--Lieb inequalities, for which a result parallel to Theorem \ref{thm:extremizableIFFgeometricAJN} holds.  However, the notion of ``geometricity" in the context of \cite{CourtadeLiu21} does not easily pass through the aforementioned limit, so it seems the simplest proof of Theorem \ref{thm:extremizableIFFgeometricAJN} is a more direct one, as given here.
\end{remark}

\section{Characterization of extremizers} \label{sec:Extremizers}

The goal of this section is to give a complete characterization of the extremizers in \eqref{eq:AJNinequality}.  In view of Theorem \ref{thm:extremizableIFFgeometricAJN}, it suffices to consider geometric instances of the AJN inequality; indeed, the extremizers of any other extremizable instance of the AJN inequality will be linear transformations of the extremizers for an   equivalent AJN-geometric datum. 

Toward this end, let $(\mathbf{c},\mathbf{d},\mathbf{B})$ be AJN-geometric, and regard $(E_i)_{i=1}^k$ and $(E^j)_{j=1}^m$ as subspaces of $E_0$, as in the discussion preceding Corollary \ref{cor:AJNgeoAlt}.  We now extend definitions found in Valdimarsson \cite{valdimarsson08} to the present setting.  A nonzero subspace $K \subset E_0$ is said to be {\bf independent} if it can be written as
$$
K = E_i \cap \bigcap_{j=1}^m V_j,
$$
for some $i\in \{1,\dots, k\}$, and each $V_j$   equal to $E^j$ or ${E^j}^{\perp}$ (the latter equal to the orthogonal complement of $E^j$ in $E_0$). Each independent subspace  is contained in some $E_i$, and distinct independent subspaces are orthogonal by construction. So, if $K^i_1, \dots, K^i_{n_i}$ is an enumeration of independent subspaces of $E_i$, then we can uniquely decompose
\begin{align}
E_i = K^i_{0}\oplus K^i_1 \oplus \cdots \oplus K^i_{n_i}, \label{eq:decomposeEi}
\end{align}
where $K^i_{0}$ is defined to be the orthogonal complement of $\oplus_{\ell=1}^{n_i}K^i_{\ell}$ in $E_i$.  Now, we can uniquely define the {\bf dependent} subspace $K_{dep}$ as the product-form subspace
\begin{align}
K_{dep}:= \oplus_{i=1}^k K^i_{0}. \label{eq:defKdep}
\end{align}
\begin{proposition} \label{prop:criticalDecomp}
If $K_{dep}$ is nonzero, there is an orthogonal decomposition
\begin{align}
K_{dep} = \oplus_{\ell=1}^n K^\ell_{dep}, \label{maxCriticalDecomp}
\end{align}
where each $K^{\ell}_{dep}$ is critical   for the datum $(\mathbf{c},\mathbf{d},\mathbf{B})$.  
\end{proposition}
A decomposition of the form \eqref{maxCriticalDecomp} is said to be  a {\bf critical decomposition}; we remark that critical decompositions are not necessarily unique.   Together with Theorem \ref{thm:extremizableIFFgeometricAJN}, the following completely characterizes the extremizers in the AJN inequality \eqref{eq:AJNinequality}.  In the statement, we let $\Pi_V:E_0 \to E_0$ denote the orthogonal projection onto the indicated subspace $V$. 

\begin{theorem}\label{thm:AJNrigidity}
Let $(\mathbf{c},\mathbf{d},\mathbf{B})$ be AJN-geometric, and decompose each $E_i$ as in \eqref{eq:decomposeEi}.  Independent \ $X_i \sim \mathcal{P}(E_i)$, $1\leq i \leq k$ and $X=(X_1,\dots, X_k)$ satisfy  \eqref{eq:AJNinequality} with equality iff
\begin{enumerate}[(i)]
\item $\Pi_{K^i_0}(X), \dots, \Pi_{K^i_{n_i}}(X)$ are independent for each $1\leq i \leq k$; and 
\item   there is a critical decomposition $K_{dep} = \oplus_{\ell=1}^n K^\ell_{dep}$ such that $\Pi_{K^1_{dep}}(X)$, \dots, $\Pi_{K^n_{dep}}(X)$ are independent isotropic Gaussians on their respective subspaces.%
\end{enumerate}
\end{theorem}
In words, (i) says that each random vector $X_i$ splits into independent factors on the orthogonal decomposition of $E_i$ given by \eqref{eq:decomposeEi}.  Condition (ii) tells us that the factor of $X$ supported on $K_{dep}$ is Gaussian with  $\ocov(\Pi_{K_{dep}} (X)) = \sum_{\ell=1}^n \sigma_{\ell}^2 \Pi_{K^\ell_{dep}}$,   for some critical decomposition \eqref{maxCriticalDecomp} and choice of variances $(\sigma_{\ell}^2)_{\ell=1}^n$.  In effect, this links the covariances of the Gaussian factors of the $X_i$'s.

\begin{remark}
In the case of $k=1$, the above characterization of extremizers is compatible with that articulated by Valdimarsson for the functional Brascamp--Lieb inequalities \cite{valdimarsson08}.   As noted in Remark \ref{rmk:FRBLrelated}, the AJN inequality is formally implied by the  Euclidean forward-reverse Brascamp--Lieb inequalities.  A characterization of extremizers for the latter remains unknown at the moment, but will necessarily involve a new ingredient of log-concavity  (since, e.g., the Pr\'ekopa--Leindler inequality is realized as a special case, and the extremizers are log-concave \cite{Dubuc}).  
\end{remark}

Before giving the proof, let us consider a few quick examples to demonstrate the result. 

\begin{example}
Consider the Shannon--Stam inequality on $E_1 = E_2 = \mathbb{R}^n$ with $\lambda\in (0,1)$, stated as
$$
\lambda h(X_1) + (1-\lambda)h(X_2) \leq h(\lambda^{1/2}X_1+ (1-\lambda)^{1/2}X_2),
$$
for independent $X_1,X_2$ with finite entropies and second moments.  There are no independent subspaces, and every maximal critical decomposition  of $K_{dep} = E_0 = \mathbb{R}^n\oplus\mathbb{R}^n $ can be written as
$$
\mathbb{R}^n\oplus\mathbb{R}^n = \bigoplus_{\ell=1}^n (\ospan\{e_{\ell}\}\oplus \ospan\{e_{\ell}\}),
$$
with $(e_{\ell})_{\ell=1}^n$ an orthonormal basis of $\mathbb{R}^n$.  Thus, (ii) is equivalent to the assertion that $X_1$ and $X_2$ must be Gaussian, with identical covariances.
\end{example}
\begin{example}
In the toy inequality \eqref{simpleAJN}, the subspace on which $X$ is supported is the only independent subspace.  So, if equality is achieved in  \eqref{simpleAJN}, then condition (i) of the theorem tells us that $X$ and $Y$ must be independent; and condition (ii) implies that $Y$ and $Z$ are Gaussian with identical covariances, as in the previous example.   %
\end{example}

\begin{example}
The Zamir--Feder inequality \cite{ZamirFeder} can be  stated as follows (see, e.g., \cite{Rioul}).  If a matrix $B \in \mathbb{R}^{k\times n}$ satisfying $B B^T = \id_{\mathbb{R}^n}$ has columns $(b_i)_{i=1}^k\subset \mathbb{R}^n$, then any random vector $X = (X_1, \dots, X_k) \in \mathcal{P}(\mathbb{R}^k)$ with independent coordinates satisfies
\begin{align}
h(B X) \geq \sum_{i=1}^k |b_i|^2 h(X_i). \label{eq:ZF}
\end{align}
Observe that this is a geometric instance of the AJN inequality, with $B_1 = B$, $d_1=1$, and $c_i = |b_i|^2$.  Letting $(e_i)_{i=1}^k$ denote the natural basis for $\mathbb{R}^k$, it follows by definitions that any independent subspace must be equal to $\ospan\{e_i\}$ for some $1\leq i \leq k$, and $\ospan\{e_i\}$ is an independent subspace iff $e_i \in \ker(B)\cup \ker(B)^{\perp}$.  
Hence,   any   $X \in \mathcal{P}(\mathbb{R}^k)$ with independent coordinates  meeting \eqref{eq:ZF} with equality has the following form:
\begin{enumerate}
\item If $e_i \in \ker(B)\cup \ker(B)^{\perp}$, then $X_i$ can have any distribution in $\mathcal{P}(\mathbb{R})$. 
\item Otherwise, $X_i$ is Gaussian. 
\end{enumerate}
Observe that $e_i \in \ker(B) \Leftrightarrow b_i = 0$; in this case, coordinate $X_i$ is not present in \eqref{eq:ZF}.  If $e_i \in \ker(B)^{\perp}$, then $X_i$ is recoverable from $BX$ in the sense that there exists $u\in \mathbb{R}^n$ such that $u^T BX = X_i$.  Hence, we might say that the extremizers in \eqref{eq:ZF} are characterized by all present non-recoverable components being Gaussian.  This is precisely the statement given by Rioul and Zamir in their recent work \cite[Theorem 1]{RioulZamir}, which gave the first characterization of   extremizers in the Zamir--Feder inequality.
\end{example}

To give an application that yields a new result, consider the following inequality     proposed in \cite{anantharam2019unifying}:
\begin{align}
c_1 h(Z_1,Z_2) + c_2 h(Y) \leq   h(Z_1+Y,Z_2+Y) + d_2 h(Z_1) + d_3   h(Z_2) + C_g, \label{eq:AJNex}
\end{align}
where the $Z_1,Z_2,Y$ are random variables with $(Z_1,Z_2)$ independent of $Y$, and all coefficients are assumed to be strictly positive.  An immediate consequence of Theorem \ref{thm:AJNinequality} is that the sharp constant $C_g$ can be computed by considering only Gaussians, and conditions on the coefficients $\mathbf{c}, \mathbf{d}$ ensuring finiteness of $C_g$ can be deduced from \eqref{eq:ScalingCond} and \eqref{eq:DimCond}.    Using Theorem \ref{thm:AJNrigidity}, we can further conclude that if $\mathbf{c}$ and $\mathbf{d}$ are such that \eqref{eq:AJNex} is extremizable, then it admits  {only} Gaussian extremizers. 

To see that this is the case, let $(\mathbf{c},\mathbf{d},\mathbf{B})$ denote the datum corresponding to \eqref{eq:AJNex}.  In matrix notation with respect to the natural choice of basis, we have 
$$
B_1 = \begin{bmatrix}
1 & 0 & 1\\
0 & 1 & 1
\end{bmatrix}, ~~B_2 = \begin{bmatrix}
1 & 0 & 0
\end{bmatrix}, ~~B_3 = \begin{bmatrix}
0 & 1 & 0
\end{bmatrix}.
$$
Assuming $(\mathbf{c},\mathbf{d},\mathbf{B})$ is extremizable, let $C$ and $(A_j)_{j=1}^3$ be the matrices in \eqref{eq:equivalentDatum} that transform $(\mathbf{c},\mathbf{d},\mathbf{B})$ to an AJN-geometric datum $(\mathbf{c},\mathbf{d},\mathbf{B'})$.  By rescaling, we can assume without loss of generality that $C = \diag(C_1,1)$, where $C_1$ is an invertible $2\times 2$ matrix. In order to show  \eqref{eq:AJNex}   admits only Gaussian extremizers, we need to show that  $(\mathbf{c},\mathbf{d},\mathbf{B'})$ admits no independent subspaces. To do this, we will show the stronger claim that 
$$
\bigcap_{j=1}^3 V_j = \{0\}
$$
for any choice of $V_j$ equal to $E^j$ or ${E^j}^{\perp}$, where we identify $E^j = \operatorname{col}(C^{-T}B_j^T A_j^{-T}) = \operatorname{col}(C^{-T}B_j^T)$, with $\operatorname{col}(\cdot)$ denoting the columnspace of its argument.  Explicitly, we have
\begin{align*}
E^1 =  \operatorname{col}\left( \begin{bmatrix} \,C_1^{-T} \,\\ 1~~~1\end{bmatrix}\right) , ~~E^2 = \operatorname{col}\left( \begin{bmatrix} \,C_1^{-T} \begin{bmatrix} 1 \\ 0\end{bmatrix} \,\\ 0\end{bmatrix}\right), ~~E^3 = \operatorname{col}\left( \begin{bmatrix} \,C_1^{-T} \begin{bmatrix} 0 \\ 1\end{bmatrix} \,\\ 0\end{bmatrix}\right).
\end{align*}
Direct computation shows 
\begin{align*}
{E^1}^{\perp} =  \operatorname{col}\left( \begin{bmatrix} \,C_1 \begin{bmatrix} 1 \\ 1\end{bmatrix}  \,\\ -1\end{bmatrix}\right) , ~~{E^2}^{\perp} = \operatorname{col}\left( \begin{bmatrix} \begin{matrix} 0 \\ 0\end{matrix} & C_1 \begin{bmatrix} 0 \\ 1\end{bmatrix} \,\\ 1 & 0\end{bmatrix}\right), ~~{E^3}^{\perp} = \operatorname{col}\left( \begin{bmatrix} \begin{matrix} 0 \\ 0\end{matrix} & C_1 \begin{bmatrix} 1 \\ 0\end{bmatrix} \,\\ 1 & 0\end{bmatrix}\right).
\end{align*}
The problem now reduces to casework.  By inspection, we   have ${E^1}^{\perp}  \cap E^2 = {E^1}^{\perp}  \cap E^3 = \{0\}$.  Next, since
  $C_1$ is invertible, we have $E^2\cap E^3 = \{0\}$, and  it similarly follows that $E^1 \cap E^2 = E^1 \cap E^3 = {E^{1}}^{\perp} \cap  {E^{2}}^{\perp} = \{0\}$. It only remains to show that $E^1  \cap  {E^{2}}^{\perp} \cap  {E^{3}}^{\perp} = \{0\}$.  To this end, invertibility of $C_1$ allows us to write
$$
{E^2}^{\perp} \cap {E^3}^{\perp} = \operatorname{col}\left( \begin{bmatrix}  0 \\0 \\1\end{bmatrix}\right).
$$
 However, the only vector in $E^1$ that is zero in the first two components is the all-zero vector (again, by invertibility of $C_1$), so it follows that $E^1  \cap  {E^{2}}^{\perp} \cap  {E^{3}}^{\perp} = \{0\}$, and we conclude that the datum $(\mathbf{c},\mathbf{d},\mathbf{B'})$ admits no independent subspaces. 

Although the above shows   \eqref{eq:AJNex} can only admit Gaussian extremizers, it does not tell us whether any exist, or their structure if they do.  This is, however,  the content of Theorem \ref{thm:extremizableIFFgeometricAJN}.  Namely, the covariances of Gaussian extremizers are characterized completely by solutions $K$ to \eqref{eq:AJNExtremizabilityCond} for the datum $(\mathbf{c},\mathbf{d},\mathbf{B})$; see Remark \ref{rmk:Kstructure}.    This emphasizes the complementary nature of Theorems \ref{thm:AJNrigidity} and \ref{thm:extremizableIFFgeometricAJN}.

\subsection{Proof of Theorem \ref{thm:AJNrigidity}}
The remainder of this section is dedicated to the proof of Theorem \ref{thm:AJNrigidity}.  We  establish the assertion of sufficiency first, and necessity second.   The assumption that the datum $(\mathbf{c},\mathbf{d},\mathbf{B})$ is AJN-geometric prevails throughout.  Accordingly we will regard $E^j$ as a subspace of $E_0$, with $\Pi_{E^j} = B_j^T B_j$ denoting the orthogonal projection onto  $E^j$.

\begin{lemma} \label{lem:SubspaceFacts} 
Let the notation of \eqref{eq:decomposeEi} and \eqref{eq:defKdep} prevail. For each $1\leq j\leq m$, we have the orthogonal decomposition
\begin{align}
E^j =  (\Pi_{E^j} K_{dep})  \oplus   \left( \bigoplus_{i =1}^k \bigoplus_{\substack{ 1\leq \ell \leq n_i:\\ K^i_{\ell} \subset E^j} }  K^i_{\ell} \right).\label{eq:decompEj}
\end{align}
Moreover, for any critical decomposition  $K_{dep} = \oplus_{\ell=1}^n K^\ell_{dep}$, we have the orthogonal decomposition 
\begin{align}
\Pi_{E^j} K_{dep} = \oplus_{\ell=1}^n \Pi_{E^j} K^{\ell}_{dep}. \label{eq:decompBjKdep}
\end{align}
\end{lemma}

\begin{proof}[Proof of Proposition \ref{prop:criticalDecomp} and Lemma \ref{lem:SubspaceFacts}]
  We first note that $\Pi_{E^j} K_{dep}$ is orthogonal to $\Pi_{E^j} K$, for any independent subspace $K$.  Indeed, by definition of an independent subspace, we either have $\Pi_{E^j} K=\{0\}$ or $\Pi_{E^j} K=K$.  The former is trivially orthogonal to $\Pi_{E^j} K_{dep}$, and the latter is orthogonal to $\Pi_{E^j} K_{dep}$ since $K_{dep}$ is orthogonal to $K$ by definition and $\Pi_{E^j}$ is self-adjoint.  Indeed, 
$$
(\Pi_{E^j} x)^T y = (\Pi_{E^j} x)^T y = x^T (\Pi_{E^j}y) = x^Ty = 0, ~~\forall x\in K_{dep}, y\in K. 
$$
This establishes \eqref{eq:decompEj}.

Now, using the decomposition \eqref{eq:decomposeEi} and the scaling condition \eqref{eq:ScalingCond} (which holds by AJN-geometricity), we have 
\begin{align*}
\sum_{i=1}^k c_i \sum_{\ell=0}^{n_i}\dim(K_{\ell}^i) = \sum_{i=1}^k c_i \dim(E_i) 
&= \sum_{j=1}^m d_j \dim(E^j) \\
&= \sum_{j=1}^m d_j \dim(\Pi_{E^j} K_{dep}) +  \sum_{j : K_{\ell}^i \subset E^j}^m d_j \dim( K^i_{\ell}).
\end{align*}
To summarize, 
\begin{align}
\sum_{i=1}^k c_i \sum_{\ell=0}^{n_i}\dim(K_{\ell}^i) &=  \sum_{j=1}^m d_j \dim(\Pi_{E^j}  K_{dep}) +  \sum_{j : K_{\ell}^i \subset E^j}^m d_j \dim( K^i_{\ell}). \label{eq:indScalingCond}
\end{align}
Since each independent subspace is of product form, the dimension condition \eqref{eq:DimCond} implies, for each $1\leq i \leq k$ and $1\leq \ell\leq n_i$,  
\begin{align}
 c_i  \dim(K_{\ell}^i) &\leq   \sum_{j : K_{\ell}^i \subset E^j}^m d_j \dim( K^i_{\ell}). \label{eq:indAreCritical}
\end{align}
Likewise, since $K_{dep} = \oplus_{i=1}^k K_0^i$ is of product form, \eqref{eq:DimCond} also implies
\begin{align}
\sum_{i=1}^k c_i  \dim(K_{0}^i) &\leq \sum_{j=1}^m d_j \dim(\Pi_{E^j} K_{dep}).  \label{eq:depIsCritical}
\end{align}
Comparing against \eqref{eq:indScalingCond},  we necessarily have equality in \eqref{eq:indAreCritical} and \eqref{eq:depIsCritical}, which proves  that $K_{dep}$ is critical.  Thus, there exists at least one critical decomposition of $K_{dep}$ (the trivial one), and Proposition \ref{prop:criticalDecomp} follows. 

It remains to show \eqref{eq:decompBjKdep}. By induction, it suffices to show if $K\subset E_0$ is a critical subspace, and $K = K_1 \oplus K_2$ is a critical decomposition, then $\Pi_{E^j}  K_1$ and $\Pi_{E^j}  K_2$ are orthogonal complements in $\Pi_{E^j} K$.   The proof is similar to that of \cite[Lemma 7.12]{bennett2008brascamp}. Letting $\Pi_{K_1}: E_0 \to E_0$ denote the orthogonal  projection onto $K_1$, we have that $\Pi_{E^j}\Pi_{K_1}$ is a contraction in $E_0$, so   $\Tr(\Pi_{E^j}\Pi_{K_1})\leq \dim(\Pi_{E^j}K_1)$. Since $K_1$ is critical, it is product-form by definition and therefore $\Pi_{K_1} = \sum_{i=1}^k \Pi_{E_i}\Pi_{K_1}\Pi_{E_i}$.  From \eqref{eq:AJNGeometricOperatorIneq}, this implies
\begin{align*}
\sum_{i=1}^k c_i \dim(\Pi_{E_i} K_1)&= \sum_{i=1}^k c_i \Tr(\Pi_{E_i} \Pi_{K_1})  = \sum_{j=1}^m d_j \Tr(\Pi_{E^j}\Pi_{K_1}) \leq \sum_{j=1}^m d_j\dim(\Pi_{E^j}K_1).
\end{align*}
Since $K_1$ is critical, we have equality throughout, implying $\Tr(\Pi_{E^j}\Pi_{K_1})= \dim(\Pi_{E^j}K_1)$ for each $j$.  From this, we can conclude that  $\Pi_{K_1}\Pi_{E^j}$ is an isometry from $\Pi_{E^j}K_1$ into $K_1$, and similarly $\Pi_{K_2}\Pi_{E^j}$ is an isometry from $\Pi_{E^j}K_2$ into $K_2$.  Since $K_1$ and $K_2$ are orthogonal complements in $K$, it follows that $\Pi_{E^j}  K_1$ and $\Pi_{E^j}  K_2$ are orthogonal complements in $\Pi_{E^j} K$. 
\end{proof}

\begin{proof}[Sufficiency of conditions (i)-(ii) in Theorem \ref{thm:AJNrigidity}]
Let $X_i \sim \mathcal{P}(E_i)$, $1\leq i \leq k$ be independent and satisfy (i)-(ii), and let $X=(X_1,\dots, X_k)$.   By the orthogonal decomposition \eqref{eq:decompEj} and the independence assumptions imposed by (i), we can decompose
\begin{align}
h(B_j X) = h(B_j \Pi_{K_{dep}}(X)) + \sum_{i =1}^k \sum_{\substack{ 1\leq \ell \leq n_i:\\ K^i_{\ell} \subset E^j} } h( \Pi_{K^i_{\ell}}(X_i) ),\label{eq:AJNstep1}
\end{align}
where all entropies are computed with respect to the subspace being projected upon.  
In the proof of  Lemma \ref{lem:SubspaceFacts}, we found \eqref{eq:indAreCritical} was met with equality.  So, whenever $E_i$ contains an independent subspace (i.e., $n_i\geq 1$), we have 
$$
c_i  =   \sum_{j : K_{\ell}^i \subset E^j}^m d_j.
$$
Now, using the decomposition \eqref{eq:decomposeEi} and the independence assumptions imposed by (i), an application of the above identity followed by \eqref{eq:AJNstep1} reveals
\begin{align*}
\sum_{i=1}^k c_i h(X_i) &= \sum_{i=1}^k \sum_{\ell=0}^{n_i} c_i h(\Pi_{K_{\ell}^i}(X_i)) \\
&= \sum_{i=1}^k   c_i h(\Pi_{K_{0}^i}(X_i))    +    \sum_{i=1}^k   \sum_{\ell=1}^{n_i}  \sum_{j : K_{\ell}^i \subset E^j}^m d_j h(\Pi_{K_{\ell}^i}(X_i)) \\
&=\sum_{i=1}^k   c_i h(\Pi_{K_{0}^i}(X_i))    +   \sum_{j=1}^m d_j  \sum_{i=1}^k   \sum_{\substack{ 1\leq \ell \leq n_i:\\ K^i_{\ell} \subset E^j} } h(\Pi_{K_{\ell}^i}(X_i)) \\
&=\sum_{i=1}^k   c_i h(\Pi_{K_{0}^i}(X_i))    +   \sum_{j=1}^m d_j  \left( h(B_j X) - h(B_j \Pi_{K_{dep}}(X))  \right).
\end{align*}
 In summary, 
\begin{align}
\sum_{i=1}^k c_i h(X_i) -\sum_{j=1}^m d_j  h(B_j X) = \sum_{i=1}^k   c_i h(\Pi_{K_{0}^i}(X_i))   - \sum_{j=1}^m d_j  h(B_j \Pi_{K_{dep}}(X)), \label{eq:OnlyDepSubspaceLeft}
\end{align}
where any entropies over the trivial subspace $\{0\}$ are to be neglected. 

It remains to show the RHS is zero.  By (ii) and translation invariance of entropy, we can assume that $\Pi_{K^{\ell}_{dep}}(X)\sim N(0,\sigma_{\ell}^2 \id_{K^{\ell}_{dep}})$ for each $1\leq \ell\leq n$.  Using the independence assumption in (ii) and the decomposition \eqref{eq:decompBjKdep}, we can express 
$$
h(B_j \Pi_{K_{dep}}(X)) = \sum_{\ell=1}^{n} \frac{\dim(B_j K_{dep}^{\ell})}{2}\log(2\pi e \sigma_{\ell}^2).
$$
Since each $K_{dep}^{\ell}$ is critical by definition, we have
\begin{align*}
\sum_{j=1}^m d_j h(B_j \Pi_{K_{dep}}(X)) &= \sum_{\ell=1}^n  \frac{1}{2} \log(2\pi e \sigma_{\ell}^2) \sum_{j=1}^m d_j \dim(B_j K_{dep}^{\ell})\\
&= \sum_{\ell=1}^n  \frac{1}{2} \log(2\pi e \sigma_{\ell}^2) \sum_{i=1}^k c_i \dim(\pi_{E_i} K_{dep}^{\ell})\\
&=  \sum_{i=1}^k c_i \sum_{\ell=1}^n  \frac{\dim(\pi_{E_i} K_{dep}^{\ell})}{2} \log(2\pi e \sigma_{\ell}^2) \\
&=  \sum_{i=1}^k c_i h(\Pi_{K_{0}^i}(X_i)),
\end{align*}
where we used the independence assumption in (ii) for the last line.  Putting everything together shows
$$
\sum_{i=1}^k c_i h(X_i) =\sum_{j=1}^m d_j  h(B_j X), 
$$
so that (i) and (ii) are sufficient conditions for the $X_i$'s to be extremal, since $C_g(\mathbf{c},\mathbf{d},\mathbf{B})=0$ by Proposition \ref{prop:AJNGeomCgEquals0}. 
\end{proof}

As we turn our attention to the necessity part of Theorem \ref{thm:AJNrigidity}, we  record several  technical lemmas for convenience.  We define $\psd(E)$ to be the closure of $\pd(E)$ (i.e., the positive semidefinite symmetric linear operators on $E$).  For $A_i \in \psd(E_i)$, $1\leq i \leq k$, we define the set $\Pi(A_1,\dots, A_k) \subset \psd(E_0)$ to be the set of symmetric positive semidefinite linear maps $A: E_0 \to E_0$ satisfying
$$
\pi_{E_i}A \pi_{E_i}^T = A_i, ~~~1\leq i \leq k.
$$
In terms of matrices, this means $A \in \Pi(A_1,\dots, A_k)$ iff $A$ is a positive semidefinite matrix with diagonal blocks $A_1, \dots, A_k$.  

\begin{lemma}\label{lem:AJNrigidMatrix}
Let $(\mathbf{c},\mathbf{d},\mathbf{B})$ be AJN-geometric, and $A_i \in \psd(E_i)$, $1\leq i \leq k$.  For any $A\in \Pi(A_1,\dots, A_k)$, we have 
\begin{align}
\sum_{i=1}^k c_i \Tr\left( ( A_i   - \id_{E_i})^2 \right)  \geq \sum_{j=1}^m d_j \Tr\left( (  (B_j A^2 B_j^T)^{1/2}   - \id_{E^j} )^2 \right), \label{eq:AJNmatrixIneqLemma}
\end{align}
with equality if and only if $(\id_{E_0}- \Pi_{E^j})A \Pi_{E^j} =0$ for each $1\leq j \leq m$. 
\end{lemma}
\begin{proof}
Using the block-diagonal structure of $A$ and the definition of AJN-geometricity, we have
\begin{align*}
\sum_{i=1}^k c_i \Tr\left( ( A_i   - \id_{E_i})^2 \right)  &=\sum_{j=1}^m d_j \Tr( B_j(A - \id_{E_0})^2 B_j )\\
&=\sum_{j=1}^m d_j \Tr( B_j A^2B_j^T - 2B_j A B_j^T   + \id_{E^j} )\\
&\geq \sum_{j=1}^m d_j \Tr( B_j A^2B_j^T - 2(B_j A^2 B_j^T)^{1/2}   + \id_{E^j} )\\
&= \sum_{j=1}^m d_j \Tr\left( (  (B_j A^2 B_j^T)^{1/2}   - \id_{E^j} )^2 \right),
\end{align*}
where the  inequality follows because  square root is operator monotone.  More precisely, AJN-geometricity implies
$$
(B_j A B_j^T)^2 =B_j A B_j^T B_j A B_j^T    \leq B_j A^2 B_j^T,
$$
so that operator monotonicity of square root gives $B_j A B_j^T \leq (B_j A^2 B_j^T)^{1/2}$.  Equality in \eqref{eq:AJNmatrixIneqLemma} is therefore equivalent to equality above, which can be rewritten as
$$
B_j A (\id_{E_0}-B_j^T B_j) A B_j^T  =0 ~~\Leftrightarrow~~(\id_{E_0}- \Pi_{E^j})A \Pi_{E^j} =0.
$$
\end{proof}

The following is due to \cite{EldanMikulincer}; we sketch the proof for completeness. 
\begin{lemma}\label{lem:EntUpperBoundGamma}  Fix a Euclidean space $E$.  
Consider a filtered probability space carrying an $E$-valued Brownian motion $(W_t)_{\geq 0}$, and let $(F_t)_{\geq 0}$  be an adapted process taking values in $\pd(E)$.  If $\int_{0}^1 F_t dW_t \sim \mu$, then 
$$
D(\mu \|\gamma_E)\leq \frac{1}{2} \int_0^1 \frac{  \EE    \Tr\left( (  F_t   - \id_{E} )^2 \right) }{1-t}dt .
$$
\end{lemma}
\begin{proof}
Define the drift 
$$u_t = \int_{0}^t \frac{F_s-\id_E}{1-s} dW_s.$$
  We claim that $W_1 + \int_{0}^1 u_t dt \sim \mu$.  To see this, write
\begin{align*}
\int_{0}^1 F_t dW_t = \int_{0}^1 \id_E dW_t  + \int_{0}^1 (F_t - \id_E) dW_t &= W_1 + \int_{0}^1 \int_t^1 \frac{F_t -  \id_E}{1-t} ds dt\\
&=W_1 + \int_{0}^1 u_s ds,
\end{align*}
where we used the stochastic Fubini theorem.  Now, by Proposition \ref{prop:FollmerLowerBound} and the data processing inequality, It\^o's isometry, and Fubini's theorem, we have
\begin{align}
D(\mu\|\gamma_E) \leq \frac{1}{2}\int_{0}^1\EE |u_t|^2dt &= \frac{1}{2}\int_{0}^1 \int_{0}^t  \frac{ \EE    \Tr\left( (  F_s   - \id_{E} )^2 \right) }{(1-s)^2}  dsdt \label{eq:ItoIsometryFs}\\
&=  \frac{1}{2} \int_0^1 \frac{  \EE    \Tr\left( (  F_s   - \id_{E} )^2 \right) }{1-s}ds .\notag
\end{align}
\end{proof}

\begin{lemma}\label{lem:GrowthEstimate}
Let $(P_t)_{t\geq 0}$ be the heat semigroup, and let $X \sim \mu \in \mathcal{P}(E)$ have density $d\mu = f d\gamma_E$.  For each $0 < t < 1$, there is a constant $C$ depending only on $t$ and the second moments of $X$ such that 
$$
|\nabla \log P_{1-t}f(x)| \leq C (|x|+1), ~~~x\in E.
$$
If, moreover, $\mu$ is of the form $\mu = \nu* \gamma_E$, then 
$$
\nabla^2 \log (P_{1-t}f(x))  \in \psd(E), ~~~x\in E, 0 < t < 1. 
$$
\end{lemma}
\begin{proof}
Let $\rho$ denote the density of $X$ with respect to Lebesgue measure on $E$.  By direct calculation, we can reparametrize $P_{1-t}f$ in terms of $\rho$ as 
$$
P_{1-t} f(x) = \left( \frac{2\pi}{t}\right)^{\dim(E)/2}e^{\frac{|x|^2}{2t}} P_{\frac{1-t}{t}}   \rho (x/t)
$$
Hence, 
\begin{align}
\nabla  \log P_{1-t}f(x) =  \frac{1}{t}x + \frac{1}{t}\nabla (\log P_{\frac{1-t}{t}}   \rho)(x/t).  \label{eq:ScoreFn}
\end{align}
Regularity estimates for evolution of densities under $(P_t)_{t\geq 0}$ imply 
$$
|\nabla \log P_s \rho(x)| \leq c_s (|x|+1), ~~~ s>0
$$ 
for some finite constant $c_s$ depending only on $s$ and the second moments of $\rho$ (see, e.g., \cite[Proposition 2]{PolyanksiyWu}).  Hence, the first claim follows.   %

For the second claim, we have $\rho = P_{1}\rho_0$ for some density $\rho_0$.  Hence, by the semigroup property combined with \eqref{eq:ScoreFn}, we have
$$
\nabla^2  \log P_{1-t}f(x) =  \frac{1}{t} \id_E + \frac{1}{t^2}\nabla^2 (\log P_{\frac{1}{t}}   \rho_0)(x/t). 
$$
By a simple convexity calculation \cite[Lemma 1.3]{EldanLee}, it holds that $\nabla^2 (\log P_s g)\geq -\frac{1}{s}\id_E$ for any density $g$ and $s>0$, so 
we find 
$$
\nabla^2  \log P_{1-t}f(x) \geq \left(\frac{1}{t} - \frac{1}{t}\right)\id_E  =0 .
$$
\end{proof}

\begin{proof}[Necessity of conditions (i)-(ii) in Theorem \ref{thm:AJNrigidity}]
Let $\mu_i \in \mathcal{P}(E_i)$, $1\leq i \leq k$ satisfy
\begin{align}
\sum_{i=1}^k c_i D(\mu_i \| \gamma_{E_i}) &= \sum_{j=1}^m d_j D(B_j \sharp (\mu_1\otimes \cdots \otimes \mu_k) \| \gamma_{E^j})
\label{eq:AJNsaturation}
\end{align} 
under the prevailing assumption of AJN-geometricity; this is the same as equality in \eqref{hIneqAJNGeom}.   Without loss of generality, we can assume each $\mu_i$ is centered.   Moreover, since the extremizers of the AJN inequality are closed under convolutions (Proposition \ref{prop:AJNextremizersClosedUnderConvolutions}) and standard Gaussians are extremal in the geometric AJN inequality (Proposition \ref{prop:AJNGeomCgEquals0}), we can assume without loss of generality that each $\mu_i$ is of the form 
\begin{align}
\mu_i = \tilde{\mu}_i * \gamma_{E_i}\label{eq:regularizedMus}
\end{align}
for some extremal $\tilde{\mu}_i \in \mathcal{P}(E_i)$, $1\leq i \leq k$.   Indeed, $X \sim \otimes_{i=1}^k \mu_i$ satisfies (i)-(ii) if and only if $X+Z$ satisfies (i)-(ii) for $Z\sim \gamma_{E_0}$, independent of $X$. 

\vskip1ex

\noindent{\bf Necessity of condition (i):}  In the proof of Proposition \ref{prop:AJNGeomCgEquals0}, the sole inequality is \eqref{eq:useFollmer2}.  Hence, properties of the drift $u_t$ warrant a closer inspection; we follow the approach developed in  \cite{EldanMikulincer}.  Toward this end, let $f$ denote the density of $\mu_1\otimes \cdots \otimes \mu_k$ with respect to $\gamma_{E_0}$, and define the function
$$
u_t(x) := \nabla \log P_{1-t} f(x),~~~~x\in E_0, ~0 \leq t \leq 1,
$$
where $(P_t)_{t\geq 0}$ denotes the heat semigroup.  Define the matrix-valued function
\begin{align}
\Gamma_t(x)  := (1-t)\nabla u_t(x) +\id_{E_0}, ~~~~x\in E_0, ~0 \leq t \leq 1, \label{def:Gammat}
\end{align}
which, for each $0\leq t \leq 1$, takes the block-diagonal form $\Gamma_t = \diag(\Gamma^1_t , \dots, \Gamma^k_t)$ with $\Gamma_t^i\in \pd(E_i)$ due to the product form of the density $f$ and Lemma \ref{lem:GrowthEstimate} applied to \eqref{eq:regularizedMus}.  

Now, consider the Wiener space of continuous functions $\mathbb{W} =  \{\omega: [0,1] \to E_0;~ \omega(0) = 0\}$, equipped with the uniform norm, the Borel sets $\mathcal{B}$, and the Wiener measure $\gamma$.   Let $X_t(\omega) = \omega(t)$ be the coordinate process, and $\mathcal{F} = (\mathcal{F}_t)_{ 0\leq t \leq 1}$ be the natural filtration of $(X_t)_{0\leq t \leq 1}$.   We'll work on the filtered probability space $(\mathbb{W}, \mathcal{B},\nu, \mathcal{F})$, where $\nu$ is the Brownian bridge
$$
\frac{d\nu}{d\gamma}(\omega) := f(\omega(1)), ~~~\omega \in \mathbb{W}.
$$
By the representation theorem for Brownian bridges, we have
\begin{align}
X_t \overset{\text{law}}{=} t X + \sqrt{t(1-t)}Z, \label{eq:AJNbridge}
\end{align}
where $X \sim \mu_1\otimes \cdots \otimes \mu_k$ and $Z\sim \gamma_{E_0}$ are independent.   Writing $u_t \equiv u_t(X_t)$, the classical de Bruijn identity, parametrized with respect to the bridge \eqref{eq:AJNbridge}, gives 
\begin{align}
D(\mu_i\|\gamma_{E_i}) = \frac{1}{2} \int_{0}^1\EE | \pi_{E_i}(u_t) |^2 dt, ~~~1\leq i\leq k, \label{eq:entEqualBridge}
\end{align}
where the expectation is with respect to $\nu$.   Moreover, if we define the $\mathcal{F}_t$-adapted process $(W_t)_{0\leq t \leq 1}$ by the equation
\begin{align}
W_t := X_t - \int_{0}^t u_s(X_s) ds, ~~~~0\leq t\leq 1,\label{eq:SDEx}
\end{align}
then $(W_t)_{0\leq t \leq 1}$ is a Brownian motion  by an application of Girsanov's theorem \cite{lehec2013representation,EldanLee}.  %
Using the SDE \eqref{eq:SDEx} and the heat equation
$$\partial_t P_{1-t} f(x) = - \frac{1}{2}\Delta P_{1-t} f(x),$$ 
we can apply It\^o's formula to $u_t$ to reveal the relationship
$$
du_t = \nabla u_t dW_t = \frac{\Gamma_t - \id_{E_0}}{1-t}dW_t,
$$
with $\Gamma_t\equiv \Gamma_t(X_t)$.  Rearranging and integrating gives
 \begin{align}
\int_{0}^1 \Gamma_t dW_t = W_1 + \int_{0}^1 u_t  dt  \sim  \mu_1\otimes \cdots \otimes \mu_k.\label{GammaDist}
\end{align}
 In particular, equality in \eqref{eq:entEqualBridge} together with the computation in \eqref{eq:ItoIsometryFs} gives the following representation for the entropies in terms of the $\Gamma_t^i$ processes:
\begin{align}
D(\mu_i \| \gamma_{E_i}) &= \frac{1}{2} \int_0^1 \frac{  \EE  \Tr\left( ( \Gamma^i_t   - \id_{E_i})^2 \right)  }{1-t}dt, ~~~1\leq i\leq k.\label{eq:entEqualBridge2}
\end{align}

Next,   positive-definiteness of $\Gamma_t$ and the assumption that $B_j B_j^T = \id_{E^j}$ together justify the definition of a new process $(\tilde{W}^j_t)_{0\leq t\leq 1}$ via
$$
d\tilde{W}^j_t = (B_j \Gamma_t^2 B_j^T)^{-1/2}  B_j \Gamma_t dW_t, ~~~1\leq j \leq m.
$$
By L\'evy's characterization, this process is a Brownian motion, since it has quadratic covariation
$$
[\tilde{W}^j]_t = \int_{0}^t (B_j \Gamma_s^2 B_j^T)^{-1/2}  B_j \Gamma_s^2 B_j^T  (B_j \Gamma_s^2 B_j^T)^{-1/2} ds = t \id_{E^j}.
$$
Putting things together, observe that definitions and \eqref{GammaDist} give
$$
 \int_0^1  (B_j \Gamma_t^2 B_j^T)^{1/2}   d\tilde{W}^j_t  = B_j \int_0^1    \Gamma_t dW_t \sim B_j \sharp (\mu_1\otimes \cdots \otimes \mu_k).
$$
Thus, by  \eqref{eq:entEqualBridge2} and an application of Lemmas \ref{lem:AJNrigidMatrix} and \ref{lem:EntUpperBoundGamma}, we have
\begin{align*}
\sum_{i=1}^k c_i D(\mu_i \| \gamma_{E_i}) &= \frac{1}{2} \int_0^1 \frac{\sum_{i=1}^k c_i \EE  \Tr\left( ( \Gamma^i_t   - \id_{E_i})^2 \right)  }{1-t}dt \\
&\geq \frac{1}{2} \int_0^1 \frac{\sum_{j=1}^m d_j  \EE    \Tr\left( (  (B_j \Gamma_t^2 B_j^T)^{1/2}   - \id_{E^j} )^2 \right) }{1-t}dt\\
&\geq \sum_{j=1}^m d_j D(B_j \sharp (\mu_1\otimes \cdots \otimes \mu_k) \| \gamma_{E^j}) .
\end{align*} 
We have equality throughout due to \eqref{eq:AJNsaturation}.   Since $X_t$ has full support for each $0 < t \leq 1$ and $(t,x)\mapsto \Gamma_t(x)$ is smooth  by the regularizing properties of the heat semigroup, Lemma \ref{lem:AJNrigidMatrix} and the above equality implies that  
\begin{align}
(\id_{E_0}- \Pi_{E^j})\Gamma_t(x) \Pi_{E^j} =0, ~~x\in E_0, ~0< t<1,~1\leq j \leq m. \label{eq:GammaTEj}
\end{align}
By definition, this implies that, for each $t\in (0,1)$, we have
\begin{align*}
(\id_{E_0}- \Pi_{E^j})\ \nabla^2 \log P_{1-t} f(x)  \Pi_{E^j}  =0, ~~x\in E_0,  ~1\leq j \leq m.
\end{align*}
Since $f$ is assumed regular by virtue of \eqref{eq:regularizedMus}, the above also holds for $t=1$ by continuity of the derivatives of the heat semigroup.  Since $f = \prod_{i=1}^k f_i$ by definition, where each $f_i$ is a density on $E_i$ with respect to $\gamma_{E_i}$, the above imposes a block-diagonal structure on the Hessian of $\log f_i$, which can be summarized as 
$$
D^2(\log f_i)(x,y) = 0,  
$$
whenever $x,y$ are vectors from distinct spaces in the decomposition \eqref{eq:decomposeEi}.  This implies, for each $1\leq i \leq k$, that  the density $f_i$  has product form
\begin{align}
 f_i(x) =   \prod_{\ell=0}^{n_i} \ f_{i,{\ell}}(\Pi_{K_{\ell}^i}(x)  ), ~~x\in E_i, \label{eq:fFactors}
 \end{align}
establishing necessity of (i).

\begin{remark}
The above proof can be viewed as a modification of Eldan and Mikulincer's argument for bounding the deficit in the Shannon--Stam inequality  \cite{EldanMikulincer}, suitable for setting of the AJN inequality.  The emergence of the  factorization \eqref{eq:fFactors} is new, and results from AJN-geometricity via the matrix inequality in Lemma \ref{lem:AJNrigidMatrix}.  Although Valdimarsson's arguments in the context of the functional Brascamp--Lieb inequalities are slightly different, the same basic factorization emerges   in \cite[Lemma 13]{valdimarsson08}.  Hence, the above might be regarded as a combination of ideas from both \cite{EldanMikulincer} and \cite{valdimarsson08}.  In the next step, the Fourier analytic argument is effectively the same as that found in \cite[Lemma 14]{valdimarsson08}, with the drift $u_t$ playing the role of what Valdimarsson calls  $\nabla \log F$. 
\end{remark}

\vskip1ex

\noindent{\bf Necessity of condition (ii):}   Having established necessity of (i), the initial calculations in the proof of sufficiency hold, leading to the conclusion \eqref{eq:OnlyDepSubspaceLeft}.  The reduced datum $(\mathbf{c},\mathbf{d},\mathbf{B}_{K_{dep}})$ obtained by restricting the maps in $\mathbf{B}$ to domain $K_{dep}$ remains AJN-geometric, so without loss of generality, we can assume for simplicity that there are no independent subspaces henceforth; i.e., $K_{dep} \equiv E_0$.  As in the previous step, we let $f$ denote the density of $X\sim \mu_1\otimes \cdots \otimes \mu_k$ with respect to $\gamma_{E_0}$. 

Letting definitions from the previous step prevail, Lemma \ref{lem:GrowthEstimate} implies  that   $u_t$ has linear growth in $x$ for each $0 < t < 1$.  Hence, we are justified in taking the Fourier transform, which we denote by $\hat{u}_t$. By \eqref{eq:fFactors}, $u_t$ is additively separable in the variables $\Pi_{E_j} x$ and $(\id_{E_0}-\Pi_{E_j})x$, and therefore  $\hat{u}_t$ is supported on $H^j \cup (H^j)^{\perp}$ for each $1\leq j \leq m$ (where $H^j$ denotes the complex Hilbert space $E^j + \mathbf{i} E^j$).  Similarly, since $u_t$ is  additively separable in the variables $\pi_{E_1}(x), \dots, \pi_{E_k}(x)$, it follows that $\hat{u}_t$ is supported on $\cup_{i=1}^k H_i$ (where,  $H_i := E_i +  \mathbf{i} E_i$).  Taking intersections, we find $\hat{u}_t$ is supported on the set %
$$
 (H_1 \cup \cdots \cup H_k) \cap \bigcap_{j=1}^m (H^j \cup (H^j)^{\perp}) = \{0\},
$$
where the equality follows by the assumption that there are no independent subspaces.   A tempered distribution  with Fourier transform supported at the origin is   a polynomial \cite[p.~194]{Rudin91}, so the linear growth estimate in Lemma   \ref{lem:GrowthEstimate} implies that $x\mapsto u_t(x)$ is affine for each $0 < t < 1$.  As a consequence of its defnition, $\Gamma_t$ is therefore deterministic for each $0 < t < 1$, in the sense that $\Gamma_t(x)$ does not depend on $x$.     Using the It\^o isometry, we  conclude from the representation $\int_{0}^1 \Gamma_t d W_t \overset{\text{law}}{=} X$  that $X$ is Gaussian with covariance
$$
\Sigma := \ocov(X  ) = \int_{0}^1  (\Gamma_t)^2  dt.
$$
Note that $\Sigma$ has diagonal form
\begin{align}
\Sigma = \Pi(\Sigma_1, \dots, \Sigma_k), ~~~\Sigma_i \in \psd(E_i), 1\leq i \leq k  \label{eq:SigmaDiagonalForm}
\end{align}
due to independence of the coordinates of $X$.  

From \eqref{eq:GammaTEj}, we have   $\Pi_{E^j}  \Sigma = \Pi_{E^j} \Sigma \Pi_{E^j}$ for each $1\leq j\leq m$.    This implies that if $v\in E_0$ is an eigenvector of $\Sigma$ with eigenvalue $\lambda$, then $\Pi_{E^j} v$ is an eigenvector of $\Pi_{E^j} \Sigma \Pi_{E^j}$ with eigenvalue $\lambda$.  In particular, if we consider the spectral decomposition $\Sigma = \sigma_1^2 \Pi_{K^{1}_{dep}} + \cdots \sigma_n^2 \Pi_{K^{n}_{dep}} $ with $\sigma^2_{1}, \dots, \sigma^2_{n}$ distinct, then we have the orthogonal decomposition
\begin{align}
B_j E_0 = \oplus_{\ell=1}^n B_j K^{\ell}_{dep},\hspace{5mm}1\leq j\leq m, \label{BjKdepDecomp2}
\end{align}
where we note each $K^{\ell}_{dep}$ is product-form due to \eqref{eq:SigmaDiagonalForm}. To see that $E_0 =  \oplus_{\ell=1}^n  K^{\ell}_{dep}$ is a critical decomposition,  observe that
\begin{align}
\sum_{i=1}^k c_i h(X_i) &=\sum_{j=1}^m d_j h(B_j  X) \label{eq:extAssump}\\
&= \sum_{\ell=1}^n  \frac{1}{2} \log(2\pi e \sigma_{\ell}^2) \sum_{j=1}^m d_j \dim(B_j K_{dep}^{\ell}) \label{eq:EigenDecomp}\\
&\geq \sum_{\ell=1}^n  \frac{1}{2} \log(2\pi e \sigma_{\ell}^2) \sum_{i=1}^k c_i \dim(\pi_{E_i} K_{dep}^{\ell})\label{eq:applyDimCondAJN}\\
&=  \sum_{i=1}^k c_i \sum_{\ell=1}^n  \frac{\dim(\pi_{E_i} K_{dep}^{\ell})}{2} \log(2\pi e \sigma_{\ell}^2) =  \sum_{i=1}^k c_i h( X_i), \label{eq:DecompKdepi}
\end{align}
where \eqref{eq:extAssump} is the extremality assumption; \eqref{eq:EigenDecomp} is due to \eqref{BjKdepDecomp2} and the spectral decomposition of $\Sigma$; \eqref{eq:applyDimCondAJN} is the dimension condition \eqref{eq:DimCond}; and  \eqref{eq:DecompKdepi} follows due to the orthogonal decomposition $E_i = \oplus_{\ell=1}^n  \pi_{E_i} K_{dep}^{\ell}$ for each $1\leq i\leq k$, because each $K_{dep}^{\ell}$ is of product-form.  Since we have equality throughout, this implies $K_{dep}\equiv E_0 = \oplus_{\ell=1}^n K_{dep}^{\ell}$ is a critical decomposition, as desired.  Since  $K_{dep}^{1}, \dots, K_{dep}^{n}$ are eigenspaces of $\Sigma$, (ii) holds.  
\end{proof}

\subsection*{Acknowledgement} 
T.C.~thanks Dan Mikulincer for his  explanations of the  properties of the F\"ollmer drift and the martingale embedding used in  \cite{EldanMikulincer}.  This work was supported in part by NSF grant CCF-1750430 (CAREER).

\appendix

\section{F\"ollmer's drift}\label{sec:FollmersDrift}

The material in this appendix can be found in  \cite{lehec2013representation}, and interested readers are referred there for more details.  We summarize the required results for completeness, since it plays an important role in the proofs of Proposition \ref{prop:AJNGeomCgEquals0} and Theorem \ref{thm:AJNrigidity}.   

 For a Euclidean space $E$, let $\mathbb{W}$ denote the classical Wiener space of continuous functions $C^0([0,1], E):= \{\omega: [0,1] \to E; \omega(0) = 0\}$ equipped with the topology of uniform convergence, and the Borel $\sigma$-algebra $\mathcal{B}$. Let $\gamma$ denote the Wiener measure on $(\mathbb{W}, \mathcal{B})$.   Let $X_t : \omega \mapsto \omega(t)$ be the coordinate process, and $\mathcal{G} = (\mathcal{G}_t)_{ 0\leq t \leq 1}$ be the natural filtration of $X = (X_t)_{0\leq t \leq 1}$.   It is a fact that $\mathcal{B}$ is the $\sigma$-algebra generated by $\mathcal{G}$.

Given a filtered probability space $(\Omega, \mathcal{A}, \mathbb{P}, \mathcal{F})$, where $\mathcal{A}$ is the Borel $\sigma$-algebra of a Polish topology on $\Omega$, a {\bf drift} is   any adapted process $U : [0,1] \to E$ such that there exists $u \in L^1([0,1];E)$ satisfying 
$$
U_t = \int_{0}^{t} u_s ds, ~~0\leq t\leq 1
$$
and $\int_{0}^{1} |u_s|^2 ds <\infty$ almost surely.   By definition and Cauchy--Schwarz, any drift $U$ belongs to $\mathbb{W}$ almost surely.

A process $B = (B_t)_{t\geq 0}$ taking values in $E$ is said to be a {\bf standard Brownian motion} if it is a Brownian motion with $B_0 =0$ and $\ocov(B_1) = \id_E$.   The following is a consequence of Girsanov's theorem; it can be found in \cite[Proposition 1]{lehec2013representation}.
\begin{proposition}\label{prop:FollmerLowerBound}
Let a standard   Brownian motion $B$, taking values in $E$, be defined on a filtered probability space  $(\Omega, \mathcal{A}, \mathbb{P}, \mathcal{F})$, and let $U_t = \int_{0}^{t} u_s ds$ be a drift.  If $\nu$ is the law of the process $(B_t+  U_t)_{0\leq t\leq 1}$, then
$$
D(\nu\|\gamma)\leq \frac{1}{2}\int_{0}^{1} \EE |  u_s |^2 ds.
$$
\end{proposition}

It turns out that the upper bound given on the relative entropy above can be met with equality.  The result is due to F\"ollmer \cite{Follmer84, Follmer86}; the statement given can be found in \cite[Theorem 2]{lehec2013representation}.
\begin{proposition}[F\"ollmer's drift]\label{prop:Follmer}
Let $\nu\ll \gamma$ be a probability measure on $(\mathbb{W}, \mathcal{B})$ with $D(\nu\|\gamma)<\infty$.  There exists an adapted process $u$ such that, under $\nu$, the following holds:
\begin{enumerate}
\item The process $U_t = \int_{0}^{t} u_s ds$   is a drift. 
\item The process $B_t = X_t- U_t$  is standard Brownian motion.
\item We have $D(\nu\|\gamma) =  \frac{1}{2} \int_{0}^{1} \EE_{\nu} |  u_s |^2 ds$.
\end{enumerate}
\end{proposition}

Let $\mu \in \mathcal{P}(E)$ have density $d\nu = fd\gamma_E$.  By defining the Brownian bridge $\nu$   on $(\mathbb{W}, \mathcal{B})$ via
\begin{align}
\frac{d\nu}{d\gamma}(\omega) =f(\omega(1)), \hspace{5mm} \omega\in \mathbb{W}, \label{eq:bridgeConstruction}
\end{align}
we have $D(\mu\|\gamma_E) =  D(\nu \|\gamma)$, which follows by data processing and the observation that $X_1\sim \gamma_E$ under $\gamma$. This gives the following convenient representation  for the entropy.  For $\mu \ll \gamma_E$ with  $D(\mu\|\gamma_E) <\infty$, let $\nu$ be the bridge in \eqref{eq:bridgeConstruction}.  On the  filtered probability space $(\mathbb{W}, \mathcal{B}, \nu, \mathcal{G})$, we have
\begin{align}
D(\mu \|\gamma_E) = \min_U \frac{1}{2} \int_{0}^{1} \EE |  u_s |^2 ds, \label{FollmerMinEnergy}
\end{align}
where the minimum is over all drifts $U_t = \int_{0}^{t} u_s ds$ such that   $\mu \sim B_1 + U_1$ for a standard Brownian motion $B$ carried by  $(\mathbb{W}, \mathcal{B}, \nu, \mathcal{G})$.  Moreover, since the process $(X_t)_{0\leq t\leq 1}$ under $\nu$ is the Brownian bridge
$$
X_t \sim t X_1 + \sqrt{t(1-t)}Z,
$$
with $Z\sim \gamma_{E}$ independent of $X_1\sim \mu$, we can take expectations in Proposition \ref{prop:Follmer}(ii) to find, with the help of Fubini's theorem, that the minimum-energy process $ (u_t)_{0\leq t\leq 1}$ in \eqref{FollmerMinEnergy} satisfies 
\begin{align}
\EE[u_t]= \int_E x d\mu(x), \hspace{5mm} \mbox{a.e.~}0\leq t\leq 1. \label{eq:FollmerMean}
\end{align}   

We now record a simple application of the above, which will suit our needs.
\begin{theorem}\label{thm:existsDrifts}
Fix probability measures  $\mu_i \ll \gamma_{E_i}$ on $E_i$ satisfying $D(\mu_i\|\gamma_{E_i})<\infty$ for each $1\leq i \leq k$.   There is a filtered probability space  $(\Omega, \mathcal{A}, \mathbb{P}, \mathcal{F})$ carrying a  Brownian motion $B$ with $\Cov(B_1)=\id_{E_0}$ and a drift  $U_t = \int_0^t u_s ds$,  $u \in L^1([0,1];E_0 )$ such that, for each $1\leq i \leq k$, 
\begin{enumerate}
\item  $\mu_i \sim \pi_{E_i}(B_1+U_1) $.
\item  $D(\mu_i\|\gamma_{E_i}) = \frac{1}{2} \int_{0}^1\EE | \pi_{E_i}(u_s) |^2 ds$.
\end{enumerate}
Moreover, the processes
$$
(B^i, u^i)=  \big(\pi_{E_i}(B_t), \pi_{E_i}(u_t)\big)_{0\leq t \leq 1}, ~~1\leq i \leq k
$$
are independent.
\end{theorem}
\begin{proof}
For each $1\leq i\leq k$, let $\mathbb{W}_i = C^0([0,1]; E_i)$, $\mathcal{G}_i$ be its natural filtration,   $\mathcal{B}_i$ be the corresponding Borel $\sigma$-algebra, and $\gamma_i$ the Wiener measure.   Define measure $\nu_i\ll \gamma_i$ on $(\mathbb{W}_i,\mathcal{B}_i)$ by
$$
\frac{d\nu_i}{d\gamma_i}(\omega) =\frac{d\mu_i}{d\gamma_{E_i}}(\omega(1)), \hspace{5mm}\omega\in \mathbb{W}_i.
$$
By   Proposition \ref{prop:Follmer} and the subsequent discussion, there exists a drift $U_t^i= \int_{0}^{t} u^i_s ds$ and a standard Brownian motion $B^i$, both carried on $(\mathbb{W}_i, \mathcal{B}_i, \nu_i, \mathcal{G}_i)$, such that $\mu_i \sim B^i_1 + U^i_1$ and 
$$D(\mu_i\|\gamma_{E_i}) =D(\nu_i\|\gamma_i) = \frac{1}{2}\int_{0}^1\EE | u^i_s |^2 ds, \hspace{5mm}1\leq i\leq k.$$
Now,  put everything together on the product space $\Omega = \prod_{i=1}^k (\mathbb{W}_i\times \mathbb{W}_i)$ equipped with its natural filtration, the Borel sets, and the product measure $\mathbb{P} = \otimes_{i=1}^k  P_{B^i U^i}$.
\end{proof}

\end{document}